\documentclass[aps,reprint,nofootinbib]{revtex4-1}

\usepackage{bm}
\usepackage{graphicx}
\usepackage{amsmath}
\usepackage{amssymb}
\usepackage{amsthm}
\usepackage{latexsym}
\usepackage{dsfont}
\usepackage{enumitem}

\newsavebox{\tallguy}
\savebox{\tallguy}{\mbox{\rule{0ex}{2.25ex}}}
\newcommand{\tr}{\mbox{Tr} \, }
\newcommand{\ket}[1]{ \usebox{\tallguy} \left | #1 \right \rangle}
\newcommand{\bra}[1]{ \left \langle #1 \right | \usebox{\tallguy}}
\newcommand{\amp}[2]{
    \left \langle #1 \left | #2 \right. \right \rangle}
\newcommand{\proj}[1]{\ket{#1} \! \bra{#1}}
\newcommand{\outerprod}[2]{\ket{#1} \! \bra{#2}}

\newcommand{\hilbert}{\mathcal{H}}
\newcommand{\qubitspace}{\mathcal{Q}}
\newcommand{\tictacspace}{\mathcal{T}}

\newcommand{\absolute}[1]{\left | #1 \right |}

\newcommand{\oper}[1]{\boldsymbol{#1}}

\newcommand{\paulix}{\boldsymbol{X}}
\newcommand{\pauliy}{\boldsymbol{Y}}
\newcommand{\pauliz}{\boldsymbol{Z}}
\newcommand{\idop}{\boldsymbol{1}}
\newcommand{\sys}[1]{^{\mbox{\tiny (#1)}}}
\newcommand{\subtext}[1]{_{\mbox{\tiny #1}}}

\theoremstyle{definition}
\newtheorem{theorem}{Theorem}
\newtheorem{lemma}{Lemma}

\begin{document}

\title{Quantum meronomic frames}
\author{Austin Hulse}
\affiliation{Department of Physics, Kenyon College, Gambier, OH 43022 and\\
Department of Physics, Duke University, Durham, NC 27708}

\author{Benjamin Schumacher\footnote{Corresponding author.}}
\affiliation{Department of Physics, Kenyon College, Gambier, OH 43022}
\email[Corresponding author:  ]{schumacherb@kenyon.edu}

\begin{abstract}
Composite quantum systems can be decomposed into subsystems in many
different inequivalent ways.  We call a particular decomposition a {\em meronomic
reference frame} for the system.  We apply the ideas of quantum reference
frames to characterize meronomic frames, identify tasks that require such
frames to accomplish, and show how asymmetric quantum states can be
used to embody meronomic frame information.
\end{abstract}

\maketitle

\section{Introduction}  \label{sec-introduction}

Consider a quantum system that is a composite of two subsystems,
designated 1 and 2.  The Hilbert space for the system is
the tensor product space $\hilbert = \hilbert\sys{1} \otimes 
\hilbert\sys{2}$, so that $\dim\hilbert  = \dim\hilbert\sys{1} \cdot \dim\hilbert\sys{2}$.
Given orthonormal bases $\{ \ket{k\sys{1}} \}$ for subsystem 1
and $\{ \ket{n\sys{2}} \}$ for subsystem 2, we can form the product
basis $\{ \ket{k\sys{1}, n\sys{2}} \}$.  
A given composite state $\ket{\Psi\sys{12}}$
is a product state provided
$\Psi\sys{12}_{kn} = \psi\sys{1}_{k} \phi\sys{2}_{n}$ 
for subsystem amplitudes
$\psi\sys{1}_{k}$ and $\phi\sys{2}_{n}$.
If no such subsystem amplitudes exist, the 
state $\ket{\Psi\sys{12}}$ is entangled.

However, any given composite system can be decomposed into 
subsystems in many different ways.  We may begin with any
orthonormal basis of $\dim\hilbert$ states in $\hilbert$ and
then assign each basis state a label $\ket{ k\sys{A}, n\sys{B} }$.
This assignment defines a decomposition of the system into subsystems
A and B, which might be quite different from the original subsystems
1 and 2.  Product and entangled states of the (A,B) decomposition
can be identified as before from the $\Psi\sys{AB}_{kn}$ 
amplitudes.

A particular division of a composite system into subsystems
constitutes a kind of ``reference frame'' for the system, which
we call a {\em meronomic} reference frame.  (The word
derives from the Greek word {\em meros} meaning ``part 
of a whole''.)  As we have seen, a product basis specifies a
particular meronomic frame.  But of course, the same subsystem
decomposition is associated with many different product bases.
Thus, assigning a product basis is sufficient but not necessary to specify
a meronomic frame.  We will give a more precise characterization
of meronomic frames in Section~\ref{sec-frames} below.

To give a simple and striking example,\footnote{Suggested 
to us by Ronan Plesser.} we begin with two qubits 1 and 2,
each with standard $\{ \ket{0}, \ket{1} \}$ bases.  From the
resulting product basis, we define the maximally entangled
Bell states:
\begin{eqnarray}  \label{eq-bellstates}
	\ket{\Phi_{\pm}\sys{12}} 
		& = & \frac{1}{\sqrt{2}} \left (\ket{0\sys{1},0\sys{2}} \pm \ket{1\sys{1},1\sys{2}} \right ) \nonumber \\
	\ket{\Psi_{\pm}\sys{12}} 
		& = & \frac{1}{\sqrt{2}} \left (\ket{0\sys{1},1\sys{2}} \pm \ket{1\sys{1},0\sys{2}} \right ).
\end{eqnarray}
The Bell states are themselves an orthonormal basis.  We can 
define an alternate meronomic frame---a decomposition of the
system into alternate qubits A and B---by 
identifying these states as a product basis in the new
decomposition.  We suppose qubit A has basis 
$\{ \ket{\Phi\sys{A}}, \ket{\Psi\sys{A}} \}$ and
qubit B has basis $\{ \ket{+\sys{B}}, \ket{-\sys{B}} \}$, 
so that
\begin{eqnarray}  \label{eq-bellproducts}
	\ket{\Phi_{+}\sys{12}} & = & \ket{\Phi\sys{A}} \otimes \ket{+\sys{B}} \nonumber \\
	\ket{\Phi_{-}\sys{12}} & = & \ket{\Phi\sys{A}} \otimes \ket{-\sys{B}} \nonumber \\
	\ket{\Psi_{+}\sys{12}} & = & \ket{\Psi\sys{A}} \otimes \ket{+\sys{B}} \nonumber \\
	\ket{\Psi_{-}\sys{12}} & = & \ket{\Psi\sys{A}} \otimes \ket{-\sys{B}} .
\end{eqnarray}
The Bell states, which are entangled in the (1,2) decomposition,
are product states in the (A,B) decomposition.
Conversely, the state
\begin{eqnarray}
	\ket{\Upsilon} & = &
		\frac{1}{2} \left ( \ket{0\sys{1}} + i\ket{1\sys{1}}\right ) \otimes \left ( \ket{0\sys{2}} - i\ket{1\sys{2}}\right )
		\nonumber \\
	& = & \frac{1}{\sqrt{2}} \left ( \ket{\Phi\sys{A}, + \sys{B}} - i \ket{\Psi\sys{A}, -\sys{B}} \right )
\end{eqnarray}
is a product state in the (1,2) frame but entangled in the (A,B) frame.

Given subsystem basis states for the A and B qubits, we can define associated Pauli
operators.  For example, $\pauliz\sys{A} = \proj{\Phi\sys{A}} - \proj{\Psi\sys{A}}$,
$\paulix\sys{B} = \ket{+\sys{B}}\bra{-\sys{B}} + \ket{-\sys{B}}\bra{+\sys{B}}$, and so on.  
This yields
\begin{equation} \label{eq-ABpauliops}
	\begin{array}{l}
	\begin{array}{l}
		\paulix\sys{A} \otimes \idop\sys{B} = \idop\sys{1} \otimes \paulix\sys{2}\\
		\pauliy\sys{A} \otimes \idop\sys{B} = \pauliz\sys{1} \otimes \pauliy\sys{2} \\ 
		\pauliz\sys{A} \otimes \idop\sys{B} = \pauliz\sys{1} \otimes \pauliz\sys{2} 
	\end{array}\\ \quad \\
	\begin{array}{l}
		\idop\sys{A} \otimes \paulix\sys{B} = \pauliz\sys{1} \otimes \idop\sys{2} \\
		\idop\sys{A} \otimes \pauliy\sys{B} = - \pauliy\sys{1} \otimes \paulix\sys{2} \\
		\idop\sys{A} \otimes \pauliz\sys{B} = \paulix\sys{1} \otimes \paulix \sys{2} 
	\end{array}
	\end{array}
\end{equation}
It is easy to confirm that these have the usual algebraic relations for 
Pauli operators on a pair of distinct qubits.  In fact, the specification of 
these operators defines the qubit decomposition\cite{zll}.

When we shift to a new meronomic frame, product states may 
become entangled and {\em vice versa}.  
But why would we choose to make such a shift?  One reason is that
the dynamics of a composite system may appear simpler in one frame than in
another.  Suppose qubits 1 and 2 represent some easily identified physical systems,
such as a pair of particle spins.  These spins interact according to a Hamiltonian
\begin{equation} \label{eq-12hamiltonian}
	\oper{H} = \alpha \, \pauliz\sys{1} \otimes \pauliz\sys{2} 
		+ \beta \, \paulix\sys{1} \otimes \paulix\sys{2} , 
\end{equation}
for some parameters $\alpha$ and $\beta$.  From Equation~\ref{eq-ABpauliops}
we recognize that this can also be written
\begin{equation}  \label{eq-ABhamiltonian}
	\oper{H} = \alpha \, \pauliz\sys{A} \otimes \idop\sys{B} 
			+ \beta \, \idop\sys{A} \otimes \pauliz\sys{B} .
\end{equation}
In the (A,B) meronomic frame, the overall Hamiltonian is the sum 
of independent subsystem Hamiltonians.  That is, 
although the spins 1 and 2 are interacting, we may equally
well regard the system as a composite of non-interacting qubits 
A and B.

This general idea is familiar from textbook quantum mechanics
\cite{qpsi}.
In a system of two interacting particles (such as a hydrogen
atom), the Hamiltonian includes kinetic energy terms for
both particles and an interaction potential:
\begin{equation}
	\oper{H} = \frac{1}{2 m_{1}} \oper{p}_{1}^2 + \frac{1}{2 m_{2}} \oper{p}_{2}^2 
		+ U(\vec{\oper{r}}_{1} - \vec{\oper{r}}_{2}) .
\end{equation}
We now introduce center of mass and relative position coordinates
\begin{equation}
	\vec{\oper{R}}\subtext{CM} = \frac{1}{M} \left ( m_{1} \vec{\oper{r}}_{1}  + m_{2} \vec{\oper{r}}_{2} \right )
	\quad \mbox{and} \quad \vec{\oper{r}}\subtext{rel} = \vec{\oper{r}}_{1} - \vec{\oper{r}}_{2} ,
\end{equation}
so that the system Hamiltonian becomes
\begin{equation}
	\oper{H} = \frac{1}{2M} \oper{P}\subtext{CM}^2 + \frac{1}{2\mu} \oper{p}\subtext{rel}^{2} 
		+ U(\vec{\oper{r}}\subtext{rel}) ,
\end{equation}
where the total mass $M = m_{1} + m_{2}$ and the relative mass $\mu = m_{1} m_{2} / M$.  
The $\vec{\oper{R}}\subtext{CM}$ and $\vec{\oper{r}}\subtext{rel}$ degrees 
of freedom are dynamically independent.  
The energy eigenstates of the system are product states
\begin{equation}  \label{eq-Hatomproduct}
\ket{\Psi} = \ket{\psi\subtext{CM}} \otimes \ket{\phi\subtext{rel}} .
\end{equation}
For a hydrogen atom, these would be plane wave states
for the center of mass along with the usual bound and unbound states
of the electron relative to the nucleus.  But these eigenstates are
{\em entangled} states of the two original particles.  The change
of coordinates includes a change of meronomic frame for the
composite system.

The distinction between meronomic frames may be
quite subtle.  Once more we start with qubits 1 and 2, and define our new
product basis for qubits A and B via
\begin{eqnarray}
	\ket{0\sys{A}, 0\sys{B}} & = & \ket{0\sys{1}, 0\sys{2}} \nonumber \\ 
	\ket{0\sys{A}, 1\sys{B}} & = & \ket{0\sys{1}, 1\sys{2}} \nonumber \\ 
	\ket{1\sys{A}, 0\sys{B}} & = & \ket{1\sys{1}, 0\sys{2}} \nonumber \\ 
	\ket{1\sys{A}, 1\sys{B}} & = & e^{i \theta} \ket{1\sys{1}, 1\sys{2}} , \label{eq-thetaframe}
\end{eqnarray}
where $\theta \in [0,2 \pi)$.  This single change of phase looks innocent, 
but it is not.  Even though all of the (A,B) basis states 
are product {\em states} in the (1,2) decomposition, the collection does not
constitute a (1,2) product {\em basis}---a basis assembled from tensor products
of independent subsystem basis states---unless $\theta = 0$.
The (A,B) decomposition is a different meronomic frame.
We can see this by considering the state
\begin{eqnarray}
	\ket{\Psi} & = & \left (\frac{1}{\sqrt{2}} ( \ket{0\sys{1}} + \ket{1\sys{1}} ) \right ) 
		\nonumber \\ & & \quad
		\otimes \left (\frac{1}{\sqrt{2}} ( \ket{0\sys{2}} + \ket{1\sys{2}} ) \right ) \nonumber \\
		& = & \frac{1}{2} \left ( \ket{0\sys{1}, 0\sys{2}} +\ket{0\sys{1}, 1\sys{2}} \right .
		\nonumber \\ & & \quad
		\left . + \ket{1\sys{1}, 0\sys{2}}  + \ket{1\sys{1}, 1\sys{2}}  \right ).
\end{eqnarray}
This is obviously a (1,2) product state.  For the A and B qubits, however,
\begin{eqnarray}
	\ket{\Psi} & = & \ket{0\sys{A}} \otimes \left ( \frac{1}{2} ( \ket{0\sys{B}} + \ket{1\sys{B}} ) \right )
			\nonumber \\ & & \quad
				+  \ket{1\sys{A}} \otimes \left ( \frac{1}{2} ( \ket{0\sys{B}} + e^{-i\theta} \ket{1\sys{B}} ) \right ) ,
\end{eqnarray}
which is entangled for all $\theta \neq 0$, and maximally entangled for $\theta = \pi$.

\begin{figure} 
\begin{center}
\includegraphics[width=3in]{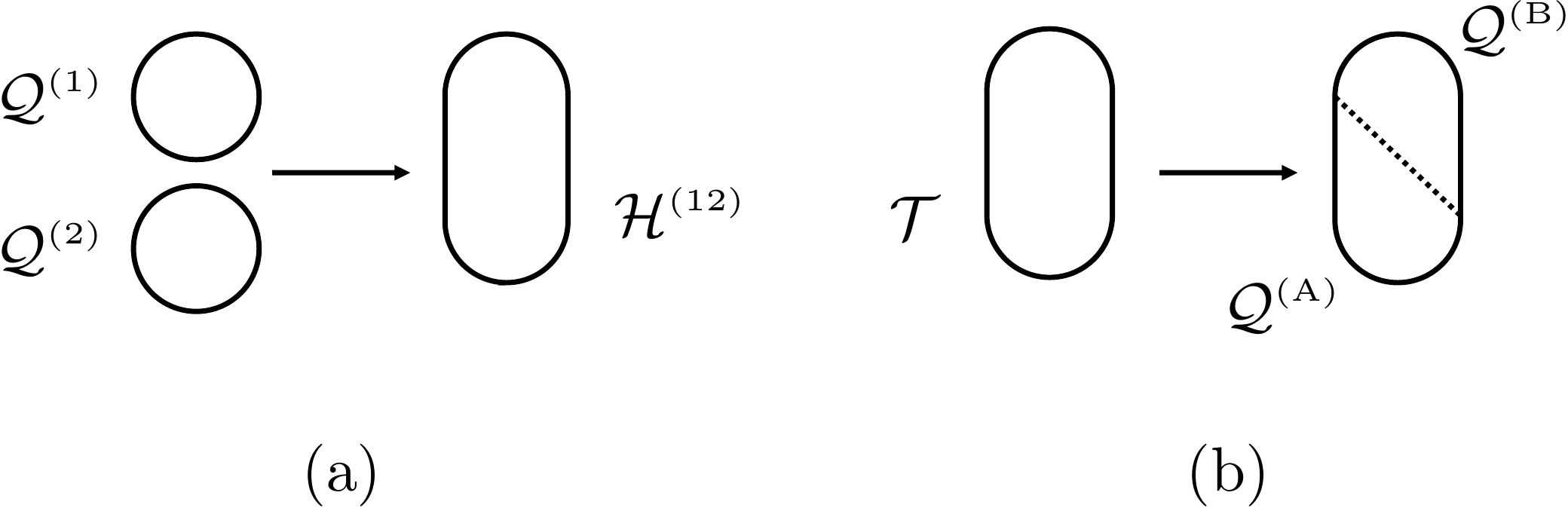}
\end{center}
\caption{In (a), we consider a composite system comprising qubits 1 and 2.
with Hilbert space $\hilbert\sys{12} = \qubitspace\sys{1} \otimes \qubitspace\sys{2}$.
In (b), a meronomic frame describes how to decompose a tictac system with
$\dim \tictacspace = 4$ into a pair of qubit subsystems A and B.}
\label{fig-tictacdiag}
\end{figure}
The general idea behind these examples is shown in Figure~\ref{fig-tictacdiag}.
Two qubits form a composite system system with Hilbert space 
$\tictacspace = \qubitspace \otimes \qubitspace$ (where $\qubitspace$ is
the qubit Hilbert space with $\dim \qubitspace = 2$).
Inspired by the oval shape of the composite in Figure~\ref{fig-tictacdiag},
we call any system with $\dim \tictacspace = 4$ a {\em tictac}.
Two qubits form a tictac, but any tictac can be decomposed into
qubits in many different ways, each corresponding to a particular
meronomic frame.

This paper investigates the concept of a meronomic frame by applying 
ideas and methods from the general theory of quantum reference frames
\cite{brs}.
The next section reviews a few ideas from that theory and describes an
example with useful analogies to our problem.
Section~\ref{sec-frames} proves a theorem that characterizes
meronomic frames.  In Section~\ref{sec-tasks} we describe a
type of operational task that requires meronomic frame information
to perform.  Section~\ref{sec-resources} shows how this meronomic
frame information can be provided in the form of a quantum state.
That is, the decomposition of one system can be encoded in the state of
other systems, which serve as a ``quantum meronomic frame''.  
Finally, in Section~\ref{sec-questions} we
make some general observations and pose a few
unanswered questions.

\section{Partial reference frames} \label{sec-partialframes}

Alice and Bob can exchange classical messages and 
quantum states, and they may perform various operations on the systems
in their possession.  They can cooperate to accomplish joint tasks.
To do so, however, they may need to share common ``reference frame''
information.  For example, if they wish to do a joint measurement 
of the $z$-component of spin on a pair of spin-1/2 particles, 
they must agree on the definition of the $z$-axis.

The idea of a reference frame is therefore closely related to the
idea of symmetry \cite{gms}.  
A symmetry group $G$---e.g., the group
of spatial rotations---connects the frames of Alice and Bob.  Each 
quantum system is equipped with a unitary representation of $G$.
A state $\ket{\psi}$ transferred from Alice to Bob is effectively changed:
$\ket{\psi} \rightarrow \oper{U}_{g} \ket{\psi}$, where $g \in G$.
In a similar way, if Bob prepares a state $\ket{\phi}$ and transfers it to Alice,
she receives $\oper{U}_{g}^{\dagger} \ket{\phi} = \oper{U}_{g^{-1}} \ket{\phi}$.
For a particular known $g \in G$, such a transformation is trivial 
and does not affect their ability to perform 
joint tasks.  In effect, Bob can simply adjust his frame to match Alice's.
However, Alice and Bob may not know
which symmetry group element relates their frames.

Here is a schematic representation of the situation:
\begin{equation} \label{eq-twobox}
	G \xrightarrow{g} \mathds{1} .
\end{equation}
Each position along the sequence represents a group of possible
frame transformations from Alice to Bob.  At the left, no frame 
information is shared, so any symmetry element in $G$ is
possible.  The arrow represents the determination of a particular
reference frame---that is, a particular group element $g$.  Once
$G$ is known by Alice and Bob, the frame group remaining is
$\mathds{1}$, which contains transformations that leave all states
invariant (but may change their overall phase).

Even if Alice and Bob do not know which element of $G$ relates
their reference frames, they can still perform {\em symmetric} cooperative tasks
with respect to $G$.  At some stage in such a task, Alice may require Bob to
do an operation that is defined in her frame.
\begin{itemize}
	\item  Alice may want Bob to produce a particular state $\ket{\phi}$
		(defined in her frame).  This is possible if $\ket{\phi}$
		is invariant with respect to the elements of $G$---that is,
		if $\oper{U}_{g} \ket{\phi} = e^{i \alpha_{g}} \ket{\phi}$ for
		all $g \in G$.
	\item  Alice may want Bob to make a measurement of
		a particular observable $\oper{A}$ (defined in her frame).  
		This is possible if $[ \oper{A}, \oper{U}_{g} ] = 0$ for all $g \in G$.
	\item  Alice may want Bob to perform a particular unitary
		transformation $\oper{V}$ (defined in her frame).
		This is possible if $[ \oper{V}, \oper{U}_{g} ] = 0$ for all $g \in G$.
\end{itemize} 

For example, suppose Alice and Bob are manipulating spin-1/2 particles,
but they lack a common set of Cartesian spatial axes.  The group $G$
of possible frame transformations is the full spin-1/2 rotation group $SU(2)$.
There are nevertheless invariant states that are the same for Alice and Bob,
such as the singlet state $\ket{\Psi_{-}}$ for a pair of spins.  Similarly,
Bob can measure the squared total angular momentum $\oper{S}^{2}$ for
the spin pair, and can reliably perform a SWAP operation on any pair of
spins.

It may happen that Alice and Bob share {\em partial} reference frame information,
so that the symmetry group of possible frame transformations is neither
all of $G$ nor the trivial group $\mathds{1}$.  In our spin-1/2 example,
Alice and Bob may share only a common $z$-axis.  

Equation~\ref{eq-zaxisframe} gives a schematic of the situation:
\begin{equation}  \label{eq-zaxisframe}
SU(2) \xrightarrow{\hat{z}} \{ \oper{R}_{z}(\theta) \} \xrightarrow{\theta} \mathds{1} .
\end{equation}
With no common frame information, the group of possible
transformations is $SU(2)$.  Once a common $z$-axis is
specified (that is, a common direction vector $\hat{z}$), 
the two frames are related by a rotation $\oper{R}_{z}(\theta)$
by an unknown angle $\theta$ about the $z$-axis.  If this
angle is specified, the two frames agree.

A partial reference frame reduces the symmetry
group from $G$ to a subgroup $G_{0}$.  The set of
partial reference frames is therefore a quotient structure
$G / G_{0}$ \cite{brs}.  This is not necessarily a group,
since $G_{0}$ need not be normal in $G$.  Both the 
$z$-axis example above and the meronomic frames 
considered below involve non-normal subgroups.  
Algebraically, a set of partial frames of this type is 
simply a collection of cosets of $G_{0}$ in $G$.
In schematic form,
\begin{equation}  \label{eq-partialframe}
G \xrightarrow{G_{0}g} G_{0} \xrightarrow{g_{0}} \mathds{1} ,
\end{equation}
where $G_{0} g$ is a (right) coset of $G_{0}$ in $G$---i.e., a 
particular partial frame.

When Alice and Bob share partial reference frame
information, they can perform cooperative tasks that
would otherwise be impossible.  In the spin-1/2 example,
given a shared $z$-axis, Bob can measure $\oper{S}_{z}$
for Alice and can prepare $\oper{S}_{z}$ eigenstates 
$\ket{\uparrow}$ and $\ket{\downarrow}$.  He can also perform 
any rotation $\oper{R}_{z}(\theta) = \exp (-i \oper{S}_{z} \theta)$ 
of a spin state about the common $z$-axis.  
Other measurements, states and operations 
remain impossible without additional reference
frame information.

In our discussion so far, we have treated reference frame
information in a wholly classical way.  Bob comes to share
a reference frame with Alice by identifying a particular element
$g$ of a symmetry group $G$ (or at least a particular coset in
$G/G_{0}$).  However, this information must
itself have a physical representation.  To put the same idea in 
different form, the reference frame is a source of asymmetry,
allowing Bob to cooperate with Alice in performing asymmetric
tasks.  Physically, this source of asymmetry must be a system
in an asymmetric state.

The theory of quantum reference frames examines how 
quantum states themselves may be used as asymmetry resources 
for asymmetric tasks.  We illustrate this by considering a special
type of problem.  Bob's reference frame differs from Alice's by some
unknown $g \in G$.  He possesses a quantum system on which
he must perform a measurement that includes the projection
$\proj{\phi}$ (as defined by Alice) as one of its effect operators.  
Unfortunately, this is not a symmetric measurement, 
since the state $\ket{\phi}$ is not invariant.
However, Bob does have at his disposal a collection of $n$ quantum
systems that have been prepared by Alice in the state $\ket{\phi}$,
forming a quantum reference frame.
How can he use this asymmetry resource to make the desired
measurement?

If $n$ were very large, Bob could use quantum tomography
\cite{leo}
to closely estimate the state $\ket{\phi}$, then use this information
to make an approximate measurement of $\proj{\phi}$ on
the system of interest.  But there is a more direct way of using
the quantum reference frame.  If the system of interest is in the
state $\ket{\psi}$ (in Alice's frame), we can write this as
\begin{equation}
	\ket{\psi} = \alpha \ket{\phi} + \beta \ket{\phi^{\perp}} ,
\end{equation}
where $\ket{\phi^{\perp}}$ is a system state orthogonal to $\ket{\phi}$.
The overall state of the $n+1$ systems is therefore
\begin{equation}
	\ket{\psi,\phi^{\otimes n}}
	= \alpha \ket{\phi^{\otimes (n+1)}} + \beta \ket{\phi^{\perp},\phi^{\otimes n}} .
\end{equation}
Bob can perform any symmetric measurement on the $n+1$ systems.  
In particular, he can measure the projection $\oper{\Pi}_{\mathcal{S}}$ 
onto the permutation-symmetric subspace $\mathcal{S}$.  (Each state
in $\mathcal{S}$ is an eigenstate of every such permutation with eigenvalue
+1.)  We can specify a basis for $\mathcal{S}$ that contains the following
two states:
\begin{eqnarray}
	\ket{\Sigma_{1}} & = & 
		\ket{\phi^{\otimes (n+1)}} \\
	\ket{\Sigma_{2}} & = & 
		\frac{1}{\sqrt{n+1}} \Big( 
		\ket{\phi^{\perp}, \phi^{\otimes n}} + 
		\ket{\phi, \phi^{\perp},\phi^{\otimes (n-1)}}  
			\nonumber \\ & & \quad
		+ \cdots + \ket{\phi^{\otimes n},\phi^{\perp}} 
		\Big )
\end{eqnarray}
so that
\begin{equation}
	\oper{\Pi}_{\mathcal{S}} \ket{\psi,\phi^{\otimes n}}
	= \alpha \ket{\Sigma_{1}} + \frac{\beta}{\sqrt{n+1}} \ket{\Sigma_{2}} .
\end{equation}
The probability that the state $\ket{\psi,\phi^{\otimes n}}$ is found to be
in $\mathcal{S}$ is
\begin{equation}
	P_{\mathcal{S}} = \bra{\psi,\phi^{\otimes n}} \oper{\Pi}_{\mathcal{S}}
		\ket{\psi,\phi^{\otimes n}} 
		= \absolute{\alpha}^{2} + \frac{\absolute{\beta}^{2}}{n+1} .
\end{equation}
This is equivalent to making a measurement on the system of interest
with an effect operator
\begin{equation} \label{eq-permsymeffect}
	\oper{E}_{\mathcal{S}} = \proj{\phi} + \frac{1}{n+1} \left (\idop - \proj{\phi} \right ).
\end{equation}
Even if $n = 1$, Equation~\ref{eq-permsymeffect} describes 
an asymmetric measurement, and so a single
qubit state can act as a reference frame.  For $n \gg 1$, the
permutation-symmetric subspace measurement yields a close 
approximation to the $\proj{\phi}$ effect operator.
If Bob has an unlimited supply of reference frame systems in state $\ket{\phi}$,
then he effectively has a classical reference frame.  But this is only a partial
reference frame.  For instance, an unlimited supply of spins in the state $\ket{\uparrow}$
can be used as a classical frame for the $z$-axis, but provides no information
about the orientation of the $x$ and $y$-axes in the plane perpendicular to $\hat{z}$.

We have seen how a single state $\ket{\phi}$ provides
a quantum resource for a partial reference frame,
leading to the construction of an approximate asymmetric 
measurement (Equation~\ref{eq-permsymeffect}).  We
will use this framework for meronomic frames
in Section~\ref{sec-resources} below.

\section{Characterizing meronomic frames} \label{sec-frames}

Subsystem decomposition determines entanglement.
If Alice and Bob share a common meronomic frame for a system,
then they will agree about which system states are product
states and which are entangled.  More abstractly, a meronomic
frame is a partial reference frame---that is, Alice and Bob
may agree on the Hilbert space tensor product structure
$\hilbert = \hilbert\sys{1} \otimes \hilbert\sys{2}$ without 
agreement on other aspects of their frames.  We must
characterize meronomic frames by identifying the subgroup of
frame transformations that leave this structure invariant.

Suppose $\hilbert = \hilbert\sys{1} \otimes \hilbert\sys{2}$.
We define the group $M$ in one of two ways:
\begin{itemize}
	\item  If $\dim\hilbert\sys{1} \neq \dim\hilbert\sys{2}$, then $M$ includes all 
		operators of the form $\oper{V}\sys{1} \otimes \oper{W}\sys{2}$,
		where $\oper{V}\sys{1}$ and $\oper{W}\sys{2}$ are
		unitary on $\hilbert\sys{1}$ and $\hilbert\sys{2}$, respectively.
	\item  If $\dim\hilbert\sys{1} = \dim\hilbert\sys{2}$, we choose a SWAP operator $\mathbb{X}$
		between the subsystems.  The group $M$ is the smallest
		group that includes all unitary $\oper{V}\sys{1} \otimes \oper{W}\sys{2}$
		and also $\mathbb{X}$.  Since
		\begin{equation}
			\mathbb{X} ( \oper{V}\sys{1} \otimes \oper{W}\sys{2}) 
				= ( \oper{W}\sys{1} \otimes \oper{V}\sys{2}) \mathbb{X},
		\end{equation}
		we can write $M = \{ (\oper{V}\sys{1} \otimes \oper{W}\sys{2}) \oper{\xi}_{i}  \}$,
		where $\oper{\xi}_{0} = \idop$ and $\oper{\xi}_{1} = \mathbb{X}$.
 \end{itemize}
The second ($\dim\hilbert\sys{1} = \dim\hilbert\sys{2}$) definition appears to depend on 
a particular isomorphism between $\hilbert\sys{1}$ and $\hilbert\sys{2}$,which we need for
SWAP operator $\mathbb{X}$, the identification of $\oper{V}\sys{1}$ with 
$\oper{V}\sys{2}$, etc.  However, the resulting group $M$ of operators
is the same regardless of the particular isomorphism (and thus the particular
$\mathbb{X}$) used to construct it.

The group $M$ defines equivalence between bipartite meronomic frames.  
That is, Alice and Bob share a meronomic frame if their reference frames are 
related by an element of $M$.  We justify this by this theorem:
\begin{theorem} \label{thm-products}
	Suppose $\oper{U}\sys{12}$ is a unitary operator on $\hilbert\sys{1} \otimes \hilbert\sys{2}$.
	The following statements are equivalent:
	\begin{enumerate}[label={\Roman*}.]
		\item  $\oper{U}\sys{12} \in M$.
		\item  For any state $\ket{\Psi\sys{12}}$, $\oper{U}\sys{12} \ket{\Psi\sys{12}}$ and 
			$\ket{\Psi\sys{12}}$ have the same Schmidt parameters.
		\item  $\oper{U}\sys{12} \ket{\Psi\sys{12}}$ is a product state if and only if
			$\ket{\Psi\sys{12}}$ is a product state.
	\end{enumerate}
\end{theorem}
\begin{proof}
(I $\Rightarrow$ II).  Suppose $\oper{U}\sys{12} \in M$, and suppose $\ket{\Psi\sys{12}}$ has
a set of Schmidt parameters $\{ \lambda_{k} \}$.  This means that
\begin{equation}
	\ket{\Psi\sys{12}} = \sum_{k} \sqrt{\lambda_{k}} \ket{k\sys{1},k\sys{2}},
\end{equation}
where $\{ \ket{k\sys{1}} \}$ and $\{ \ket{k\sys{2}} \}$ are orthonormal sets of states
in $\hilbert\sys{1}$ and $\hilbert\sys{2}$.  Unitary operators map orthonormal
states to orthonormal states, so
\begin{equation}
	\oper{V}\sys{1} \otimes \oper{W}\sys{2} \ket{\Psi\sys{12}} 
		= \sum_{k} \sqrt{\lambda_{k}} \ket{{k'}\sys{1},{k'}\sys{2}}
\end{equation}
for orthonormal sets $\{ \ket{{k'}\sys{1}} \}$ and $\{ \ket{{k'}\sys{2}} \}$.
This state has the same Schmidt parameters $\lambda_{k}$.  Furthermore, if
$\dim \hilbert\sys{1} = \dim \hilbert\sys{2}$ and we define the SWAP operator
$\mathbb{X}$ so that $\mathbb{X} \ket{k\sys{1}, l\sys{2}} = \ket{l\sys{1},k\sys{2}}$,
we have that $\mathbb{X} \ket{\Psi\sys{12}} = \ket{\Psi\sys{12}}$.
Therefore the Schmidt parameters $\lambda_{k}$ are unchanged by
the application of any local unitary transformation 
$\oper{V}\sys{1} \otimes \oper{W}\sys{2}$ or of $\mathbb{X}$, 
and thus of any combination of these operators.  This includes every
$\oper{U}\sys{12} \in M$.

(II $\Rightarrow$ III).  Product states are those states with a single
nonzero Schmidt parameter $\lambda = 1$.  Thus, condition
2 implies condition 3 as a special case.

(III $\Rightarrow$ I).  We suppose that $\oper{U}\sys{12}$ and its
inverse ${\oper{U}\sys{12}}^{\dagger}$ take product states to
product states.  Let $\{ \ket{j\sys{1},m\sys{2}} \}$
be a product basis for $\hilbert\sys{1} \otimes \hilbert\sys{2}$.
Then the product states
\begin{equation}
	\ket{\alpha_{jm}\sys{1}, \beta_{jm}\sys{2}} = \oper{U}\sys{12} \ket{j\sys{1},m\sys{2}}
\end{equation}
constitute a basis of product states, though we do not yet know that it
is a product basis.

Dropping the subsystem identifiers for simplicity, we now
consider the states $\ket{\alpha_{0m}, \beta_{0m}} = \oper{U} \ket{0, m}$.
For $m \neq n$, the superposition
\begin{equation}
	\oper{U} \left ( \frac{1}{\sqrt{2}} (\ket{0,m} + \ket{0,n}) \right ) = 
		\frac{1}{\sqrt{2}} (\ket{\alpha_{0m},\beta_{0m}} + \ket{\alpha_{0n},\beta_{0n}} )
\end{equation}
must also be a product state.  It follows that either $\absolute{\amp{\alpha_{0m}}{\alpha_{0n}}} = 1$
or $\absolute{\amp{\beta_{0m}}{\beta_{0n}}} = 1$.  These two possibilities are mutually
exclusive.  Since $\ket{\alpha_{0m},\beta_{0m}}$ and $\ket{\alpha_{0n},\beta_{0n}}$ are
orthogonal, either $\amp{\alpha_{0m}}{\alpha_{0n}} = 0$ or $\amp{\beta_{0m}}{\beta_{0n}} = 0$.
Thus, exactly one of the following conditions must hold:
\begin{enumerate}[label=({\Alph*})]
	\item $\absolute{\amp{\alpha_{0m}}{\alpha_{0n}}} = 1$ and $\amp{\beta_{0m}}{\beta_{0n}} = 0$.
	\item $\amp{\alpha_{0m}}{\alpha_{0n}} = 0$ and $\absolute{\amp{\beta_{0m}}{\beta_{0n}}} = 1$.
\end{enumerate}
In fact, the same condition (either (A) or (B)) must hold for every pair
$\ket{\alpha_{0m},\beta_{0m}}$ and $\ket{\alpha_{0n},\beta_{0n}}$.
To see this, consider three states $\ket{\alpha_{0l},\beta_{0l}}$, $\ket{\alpha_{0m},\beta_{0m}}$
and $\ket{\alpha_{0n},\beta_{0n}}$.  
If $\absolute{\amp{\alpha_{0l}}{\alpha_{0m}}} = \absolute{\amp{\alpha_{0m}}{\alpha_{0n}}} = 1$,
then $\absolute{\amp{\alpha_{0l}}{\alpha_{0n}}} = 1$, and condition (A) holds for all three pairs.
Conversely, if $\absolute{\amp{\beta_{0l}}{\beta_{0m}}} = \absolute{\amp{\beta_{0m}}{\beta_{0n}}} = 1$,
then $\absolute{\amp{\beta_{0l}}{\beta_{0n}}} = 1$, and condition (B) holds.

By considering the states $\ket{\alpha_{j0}, \beta_{k0}}$, we arrive at a corresponding
pair of conditions, exactly one of which must hold for all such states:
\begin{enumerate}[label=({\Alph*}')]
	\item $\amp{\alpha_{j0}}{\alpha_{k0}} = 0$ and $\absolute{\amp{\beta_{j0}}{\beta_{k0}}} = 1$.
	\item $\absolute{\amp{\alpha_{j0}}{\alpha_{k0}}} = 1$ and $\amp{\beta_{k0}}{\beta_{k0}} = 0$.
\end{enumerate}
These conditions are related to the previous ones, in that either (A) and (A') both hold,
or (B) and (B') both hold.  Suppose instead that both (A) and (B') are true.  The state
\begin{eqnarray}
	\ket{\Phi} & = & \oper{U} \left ( \frac{1}{\sqrt{3}} (\ket{0,0} + \ket{0,m} + \ket{j,0}) \right )
		\nonumber \\
		& = & \frac{1}{\sqrt{3}} (\ket{\alpha_{00},\beta_{00}} + \ket{\alpha_{0m},\beta_{0m}} 
		+ \ket{\alpha_{j0},\beta_{j0}})
\end{eqnarray}
must be entangled.  From condition (A) we deduce that 
$\ket{\alpha_{0m}} = e^{i \phi} \ket{\alpha_{00}}$,
and from condition (B') we know $\ket{\alpha_{j0}} = e^{i \phi'} \ket{\alpha_{00}}$,
for some complex phases $e^{i \phi}$ and $e^{i \phi'}$.
Then
\begin{equation}
	\ket{\Phi} = \ket{\alpha_{00}} \otimes \frac{1}{\sqrt{3}}  (\ket{\beta_{00}} + e^{i \phi} \ket{\beta_{0m}} 
		+ e^{i \phi'} \ket{\beta_{j0}}),
\end{equation}
which is obviously a product state.  Hence (A) and (B') are inconsistent.  A similar
argument shows that (B) and (A') are also inconsistent.  Either both (A) and (A') hold,
or both (B) and (B') hold.

If $\dim \hilbert\sys{1} \neq \dim \hilbert\sys{2}$, we know that (A) and (A') are true.
That is, if $\dim \hilbert\sys{1} < \dim \hilbert\sys{2}$, condition (B) must be false,
since $\hilbert\sys{1}$ is not large enough to contain $\dim \hilbert\sys{2}$ 
orthogonal states.  Likewise, if $\dim \hilbert\sys{1} > \dim\hilbert\sys{2}$,
condition (B') must fail.  Conditions (B) and (B') require that the two Hilbert spaces
be isomorphic, so that $\dim \hilbert\sys{1} = \dim \hilbert\sys{2}$.

We will now show that, under conditions (A) and (A'), we can write
$\oper{U} = \oper{V} \otimes \oper{W}$.  First we
define $\ket{a_{0}} \otimes \ket{b_{0}} = \ket{\alpha_{00},\beta_{00}}$,
and then define other subsystem states $\ket{a_{j}}$ and $\ket{b_{m}}$
so that
\begin{eqnarray}
	\ket{\alpha_{0m},\beta_{0m}} & = & \ket{a_{0}} \otimes \ket{b_{m}}\\
	\ket{\alpha_{j0},\beta_{j0}} & = & \ket{a_{j}} \otimes \ket{b_{0}} .
\end{eqnarray}
The states $\ket{a_{j}}$ (including $\ket{a_{0}}$) form a basis for $\hilbert\sys{1}$,
just as the $\ket{b_{m}}$ states form a basis for $\hilbert\sys{2}$.
We further observe that the state
\begin{eqnarray}
	\ket{\Psi} & = & \oper{U} \left ( \frac{1}{2} (\ket{0} + \ket{j}) \otimes (\ket{0} + \ket{m}) \right )
		\nonumber \\
		& = & \frac{1}{2} \big( \ket{a_{0},b_{0}} + \ket{a_{j},b_{0}} 
		\nonumber \\ & & \quad
			+ \ket{a_{0},b_{m}} + \ket{\alpha_{jm},\beta_{jm}} \big)
\end{eqnarray}
must be a product state, which can only occur if $\ket{\alpha_{jm},\beta_{jm}} 
= \oper{U} \ket{j,m} = \ket{a_{j},b_{m}}$.  We can therefore define
unitary subsystem operators so that $\oper{V} \ket{j} = \ket{a_{j}}$ and 
$\oper{W} \ket{m} = \ket{b_{m}}$.  Because $\oper{U} = \oper{V} \otimes \oper{W}$,
it follows that $\oper{U} \in M$.

Now imagine that $\dim \hilbert\sys{1} = \dim \hilbert\sys{2}$, and that both
conditions (B) and (B') hold.  We introduce a SWAP operator $\mathbb{X}$
for which $\mathbb{X} \ket{j,m} = \ket{m,j}$.  It is not difficult to see that
the operator $\oper{U} \mathbb{X}$ preserves product states and 
must satisfy conditions (A) and (A').  Therefore we can write 
$\oper{U} = (\oper{V} \otimes \oper{W}) \mathbb{X}$, which is
also in the group $M$.
\end{proof}

Two frames that are related by an element of $M$ will have the
same product states, and indeed the same Schmidt parameters 
for all bipartite states.  The two frames thus correspond to the same
bipartite meronomic frame.  Algebraically, let us identify the full
unitary group $U(\hilbert)$ as the
group of all possible frames.  For a general $\oper{U} \in
U(\hilbert)$, the states $\ket{\Psi}$
and $\oper{U} \ket{\Psi}$ may have very different Schmidt
parameters.  The possible meronomic frames
are just the cosets $U(\hilbert) / M$.
Two frame operators $\oper{U}_{1}$ and $\oper{U}_{2}$ are in
the same coset if they differ only by an element of $M$:
\begin{equation}
	\oper{U}_{2} = \oper{M} \oper{U}_{1}
\end{equation}
for some $\oper{M} \in M$.  Thus, two operators $\oper{U}_{1}$
and $\oper{U}_{2}$ correspond to the
same meronomic frame transformation if they change the 
Schmidt parameters of a state in the same way.

We can extend Theorem~\ref{thm-products} in various ways.
For example, the equivalence of statements 1 (that $\oper{U} \in M$)
and 3 (that $\oper{U}$ preserves product states) can be extended
to the $n$-partite case by induction.  We omit the proof of this
fact, and only remark that the group $M$ must be extended with
system-permutation operations among all subsystems having 
the same Hilbert space dimension.

For tictac systems composed of two qubits, we have a further
characterization of $M$:  two frames entail the same meronomic
frame if and only if the same states are maximally entangled
in each.  We state and prove this theorem in the appendix.

\section{Meronomic tasks} \label{sec-tasks}

If Alice and Bob share a meronomic frame for a composite system,
they should be able to perform cooperative tasks that would be
impossible absent such a shared frame.  What sort of task does
a meronomic frame enable?  Here is a simple example.  We suppose
$\dim \hilbert\sys{1} = \dim \hilbert\sys{2}=d$, so that the meronomic
group $M$ includes both local frame transformations and a SWAP
operation.  This means that Alice and Bob do not have shared 
subsystem frames or subsystem labels.
In our proposed task, Alice prepares a state
$\ket{\Psi}$ and gives it to Bob.  Bob receives $\ket{\Psi'} = \oper{U}_{g} \ket{\Psi}$
for some $\oper{U}_{g} \in M$.  His task is to apply a 
unitary $\oper{V}$ to the system that will yield a state
$\oper{V}\ket{\Psi'}$ orthogonal to $\ket{\Psi'}$.

This task could be used as part of a communication scheme.
Bob can either choose to apply $\oper{V}$ or the identity $\idop$
to the system.  If he then returns the state to Alice, she can make a measurement
in her own frame to determine which operator has been applied.
Also note that, if Alice and Bob share no common frame
information at all, the task is impossible.  If $\oper{U}_{g}$ is completely
unknown, the state $\ket{\Psi'}$ might turn out to be 
an eigenstate of $\oper{V}$, so that $\ket{\Psi'}$ and 
$\oper{V} \ket{\Psi'}$ would be indistinguishable. 

By contrast, if Alice and Bob share a meronomic frame, they can accomplish
this task.  Alice prepares $\ket{\Psi}$ as a maximally entangled
state of the two subsystems.  Therefore $\ket{\Psi'}$ is also
maximally entangled.  If Bob chooses a basis $\{ \ket{k\sys{1}} \}$
for his subsystem 1, it must be that
\begin{equation}
	\ket{\Psi'} = \frac{1}{\sqrt{d}} \sum_{k} \ket{k\sys{1},\psi_{k}\sys{2}}
\end{equation}
for an orthonormal basis $\{ \ket{\psi_{k}\sys{2}} \}$ (unknown to Bob).
Bob chooses $\oper{V} = \oper{W}\sys{1} \otimes \idop\sys{2}$, where
\begin{equation}
	\oper{W} \ket{k\sys{1}} = \ket{((k+1) \bmod d)\sys{1}} .
\end{equation}
Then
\begin{eqnarray}
	\bra{\Psi'} \oper{V} \ket{\Psi'}
	& = & \frac{1}{d} \sum_{j,k} \amp{j\sys{1}}{((k+1)\bmod d)\sys{1}} \amp{\psi_{j}\sys{2}}{\psi_{k}\sys{2}}
	\nonumber \\
	& = & \frac{1}{d} \sum_{k} \amp{\psi_{(k+1)\bmod d}\sys{2}}{\psi_{k}\sys{2}} \nonumber \\
	& = & 0.
\end{eqnarray}
(This choice of $\oper{V}$ is far from unique.  Any operator of the
form $\oper{W} \otimes \idop$ or $\idop \otimes \oper{W}$, where
$\tr \oper{W} = 0$, will work as well.)

This protocol is in fact a partial version of superdense coding
\cite{bw92,ncbook}.
Suppose the system is a tictac, and that Alice and Bob agree
on its division into a pair of qubits.
Bob only interacts with one of the qubits (a fact that makes
sense in both frames), but is nevertheless able to encode 
a one-bit message in the entangled state.  
If he and Alice share a complete frame for the tictac, then
he can do even more.  The operators $\oper{1} \otimes \oper{1}$,
$\oper{X} \otimes \oper{1}$, $\oper{Y} \otimes \oper{1}$ and
$\oper{Z} \otimes \oper{1}$ yield four states that are 
distinguishable in Alice's frame, so that Bob can encode
a two-bit message.

Now we turn to another example of a task, one that 
will reveal a good deal about meronomic frames.  
Again, our basic system is a tictac, and the frames of
Alice and Bob are related by some $\oper{U}_{g} \in M$.
The degree of entanglement of a two-qubit state is characterized
by a single independent Schmidt parameter.  That is, given
the composite state $\ket{\Phi}$ we can find bases $\{\ket{a_{0}},\ket{a_{1}} \}$
for $\qubitspace\sys{1}$ and $\{\ket{b_{0}},\ket{b_{1}} \}$
for $\qubitspace\sys{2}$ so that
\begin{equation}
	\ket{\Phi} = \sqrt{\lambda} \ket{a_0,b_0}
		+ \sqrt{1-\lambda} \ket{a_{1},b_{1}} 
\end{equation}
for some $\lambda \in [0,1/2]$.  The state $\ket{\Phi}$
is a product state when $\lambda = 0$ and maximally
entangled when $\lambda = 1/2$.

Suppose Alice prepares a tictac in the state $\ket{\Phi}$ with
Schmidt parameter $\lambda$, and then she
delivers it to Bob, who receives the state $\oper{U}_{g} \ket{\Psi}$.
This state has the same value of $\lambda$ as
Alice's original.  Consider therefore a task in which Bob determines
$\lambda$ for a state $\ket{\Phi}$ provided by Alice.  This is
plainly impossible without a shared meronomic frame, but
perhaps possible with such a shared frame.

Even with a shared meronomic frame, though, Bob cannot determine anything
about $\lambda$ from a single copy of the tictac state $\ket{\Phi}$. 
Since Alice and Bob do not know which transformation
$\oper{U}_{g} \in M$ connects their frames, Bob effectively 
receives the mixed state \cite{brs}
\begin{equation}
	\oper{\sigma} = \int_{M} dg \, \oper{U}_{g} \proj{\Phi} \oper{U}_{g}^{\dagger} .
\end{equation}
(Here we use the unique invariant measure on the compact group $M$.)
The density operator $\oper{\sigma}$ is sometimes called the ``twirl'' of
$\proj{\Psi}$ over $M$, and is itself $M$-invariant.  That is,
\begin{equation}
	\oper{\sigma} = \oper{U}_{g} \oper{\sigma} \oper{U}_{g}^{\dagger}
\end{equation}
for every $\oper{U}_{g} \in M$.  But there is only one such mixed 
state:  $\oper{\sigma} = \frac{1}{4} \idop\sys{1} \otimes \idop\sys{2}$.
Therefore, no matter what measurement Bob chooses to make on
the tictac, the measurement outcomes will have the same probabilities
regardless of Alice's original state $\ket{\Psi}$.  Bob can learn nothing
about the parameter $\lambda$.

If Alice sends many tictacs to Bob, each prepared in the same state
$\ket{\Psi}$, then Bob can do better.  He receives the state
$\left ( \oper{U}_{g} \ket{\Psi} \right )^{\otimes N}$, and by performing quantum
tomography on the tictacs he can arrive at an estimate of the
state $\oper{U}_{g} \ket{\Psi}$, and thus an estimate of $\lambda$.
This estimate will converge on the correct value as $N \rightarrow \infty$.
Thus, the task we propose is that Bob estimates the Schmidt 
parameter $\lambda$ for the tictac state $\ket{\Phi}$, given that 
Alice has prepared $N>1$ tictacs in that state.

Bob can obtain some information about $\lambda$ even if 
$N = 2$.
Recall the two-qubit singlet state
$\ket{\Psi_{-}\sys{12}}$ of Equation~\ref{eq-bellstates}.  If $\oper{V}$ is
a single-qubit unitary, then
\begin{equation}
	\oper{V}\sys{1} \otimes \oper{V}\sys{2} \ket{\Psi_{-}\sys{12}} = 
		e^{i \phi} \ket{\Psi_{-}\sys{12}}
\end{equation}
for some phase $e^{i \phi}$. 
Given tictacs 12 and 34 (comprising qubits 1 and 2, and 3 and 4,
respectively), we define
\begin{equation}
	\ket{\Lambda} = \ket{\Psi_{-}\sys{13},\Psi_{-}\sys{24}} .
\end{equation}
This is illustrated in Figure~\ref{fig-lambdastate}
\begin{figure}
\begin{center}
\includegraphics[width=2in]{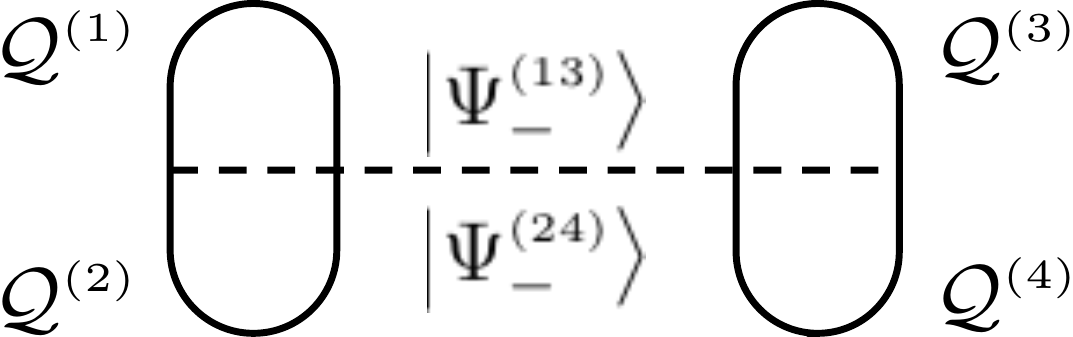}
\end{center}
\caption{The $\ket{\Lambda}$ state is a product state
	between 13 and 24 but is an entangled state of the tictacs
	12 and 34.  \label{fig-lambdastate}}
\end{figure}
For the two-tictac system 1234, any element of the meronomic group $M$
is of the form
\begin{equation}
	\oper{U}_{g} = \left ( \oper{V}\sys{1} \otimes \oper{W}\sys{2}
			\otimes \oper{V}\sys{3} \otimes \oper{W}\sys{4} \right )
			\left ( \oper{\xi}_{i}\sys{12} \otimes \oper{\xi}_{i}\sys{34} \right) ,
\end{equation}
where as before $\oper{\xi}_{i}$ is either $\idop$ or $\mathbb{X}$.
Clearly $(\oper{\xi}_{i}\sys{12} \otimes \oper{\xi}_{i}\sys{34}) \ket{\Lambda}
= \ket{\Lambda}$, and thus
\begin{eqnarray}
	\oper{U}_{g} \ket{\Lambda}
	& = & \left ( \oper{V}\sys{1} \otimes \oper{V}\sys{3} \ket{\Psi_{-}\sys{13}} \right )
		\otimes \left (\oper{W}\sys{2} \otimes \oper{W}\sys{4} \ket{\Psi_{-}\sys{24}} \right )
		\nonumber \\
	& = & e^{i \phi} \ket{\Lambda}.
\end{eqnarray}
The state $\ket{\Lambda}$ is invariant under $M$.  If Alice
prepares two tictacs in this state, Bob receives the state unchanged
except possibly in phase.  Furthermore, on any tictac pair 
he can make a symmetric measurement that has $\proj{\Lambda}$ 
as one of its effects.

Suppose therefore that Alice sends two tictacs in a duplicated
state $\ket{\Phi,\Phi}$.  If we fix standard bases
$\{ \ket{0},\ket{1} \}$ for the qubits, we can always write
$\ket{\Phi} = \oper{V} \otimes \oper{W} \ket{\Phi'}$, with
\begin{equation}
	\ket{\Phi'} = \sqrt{\lambda} \ket{00} + \sqrt{1-\lambda} \ket{11} ,
\end{equation}
so that
\begin{eqnarray}
	\ket{\Phi',\Phi'} & = & \lambda \ket{0000} + \sqrt{\lambda (1-\lambda)}
		\ket{0011} \nonumber \\ & & + \sqrt{\lambda(1-\lambda)} \ket{1100}
		+ (1-\lambda) \ket{1111} .
\end{eqnarray}
We can also write
\begin{equation}
	\ket{\Lambda} =
		\frac{1}{2} \left ( \ket{0011} - \ket{0110} - \ket{1001} + \ket{1100} \right ) .
\end{equation}
Hence, if Bob makes a symmetric measurement with effect
$\proj{\Lambda}$, the probability of this effect is
\begin{eqnarray}
	\absolute{\amp{\Lambda}{\Phi,\Phi}}^2
	& = & \absolute{\bra{\Lambda}
		\oper{V} \otimes \oper{W} \otimes \oper{V} \otimes \oper{W}
		\ket{\Phi',\Phi'}}^2 \nonumber \\
	& = & \absolute{\amp{\Lambda}{\Phi',\Phi'}}^2 \nonumber \\
	& = & \lambda (1 - \lambda). \label{eq-lambdaprob}
\end{eqnarray}
If $\ket{\Phi}$ is a product state of qubits ($\lambda = 0$), 
the $\proj{\Lambda}$ effect never occurs in the measurement.
Even a single measurement on a duplicated tictac state
$\ket{\Phi, \Phi}$ may thus reveal that $\ket{\Phi}$ is
entangled.  If Alice provides many tictac pairs in the state $\ket{\Phi,\Phi}$, 
Bob can make a statistical estimate of $P_{\Lambda} = \lambda (1-\lambda)$
and infer the Schmidt parameter $\lambda$ itself.

The state $\ket{\Lambda}$ has several remarkable properties that
are worth noting.  It, like every duplicated state $\ket{\Phi,\Phi}$,
lies in the symmetric subspace $\mbox{Sym}(\tictacspace \otimes \tictacspace)$.
This subspace has dimension 10.  The duplicated product states all
lie in a 9-dimensional subspace of $\mbox{Sym}(\tictacspace \otimes \tictacspace)$, 
and $\ket{\Lambda}$ is the unique symmetric state orthogonal to all
of them.  Furthermore, no duplicated entangled state is orthogonal
to $\ket{\Lambda}$.

\section{Quantum meronomic resources}  \label{sec-resources}

We have seen that the only tictac state that is invariant
under the meronomic group $M$ is the mixed state 
$\frac{1}{4} \idop\sys{1} \otimes \idop\sys{2}$.  This is
a special case of a more general fact.  Suppose we have
any composite system with Hilbert space 
$\hilbert = \hilbert\sys{1} \otimes \hilbert\sys{2} \otimes \cdots$.
This system has a meronomic group $M$ as defined in 
Section~\ref{sec-frames}.  If we ``twirl'' any state by the
elements of $M$, we obtain the unique $M$-invariant state
\begin{equation}
	\oper{\sigma} = \frac{1}{d} \, \idop\sys{1} \otimes \idop\sys{2} \otimes \cdots
\end{equation}
where $d = \dim \hilbert$.  Thus, no state of a single system
can convey meronomic frame information from Alice to Bob,
nor can it serve as a quantum resource for a cooperative task that
depends on a common meronomic frame.  One tictac can
indicate nothing about how a tictac should be decomposed
into qubits.

A pair of tictacs in the state $\ket{\Lambda}$, however,
does carry meronomic information.  Suppose Bob wishes
to make the $\lambda$-estimating measurement on the
duplicated state $\ket{\Phi,\Phi}$.  Alice has provided $n$ 
additional tictac pairs in the state $\ket{\Lambda}$ to serve
as a quantum meronomic resource.  This is analogous to
the situation in Section~\ref{sec-partialframes} in which 
$n$ spin states $\ket{\uparrow}$ served as a quantum 
reference frame for a measurement of $\oper{S}_{z}$ on
another spin.  Bob now proceeds exactly as he did then.
His $n+1$ tictac pairs are in the state
$\ket{\Phi,\Phi,\Lambda^{\otimes n}}$, and he measures
whether this state is in the permutation-symmetric
subspace of the $n+1$ pairs.  Equation~\ref{eq-permsymeffect}
tells us that this is equivalent to measuring the effect
\begin{equation}
	\oper{E}_{\mathcal{S}} = \proj{\Lambda} + \frac{1}{n+1} \left (\idop - \proj{\Lambda} \right ).
\end{equation}
on the tictac pair of interest.  This provides some 
meronomic frame-dependent information even for $n=1$, 
and it yields a good approximation of the desired 
measurement for $n \gg 1$.

Specifying the state $\ket{\Lambda}$ for a tictac pair
completely determines the tictac meronomic frame,
since the Schmidt parameter $\lambda$ of any tictac
state $\ket{\Phi}$ is determined by Equation~\ref{eq-lambdaprob}.
But since $\ket{\Lambda}$ is $M$-invariant, the state
provides no additional frame information beyond the 
subsystem decomposition itself.
An unlimited supply of $\ket{\Lambda}$ states is therefore
equivalent to a classical meronomic frame.

The meronomic frames of a tictac system
form a continuous set.  The product basis states in 
Equation~\ref{eq-thetaframe} define a different meronomic
frame for each value of $\theta \in [0,2\pi)$.
Therefore, by the reference-frame
version of the Wigner-Araki-Yanase theorem
\cite{ms}, 
no finite collection of systems can convey complete
meronomic frame information.  As we will now show, 
however, a single tictac pair {\em can} convey complete 
information about how subsystems are to be labeled.

We have defined the meronomic group $M$ to include
the swap operation $\mathbb{X}$ between identical
subsystems.  That is, Alice and Bob may share a common
decomposition of a tictac into a pair of qubits, but not
have a common assignment of labels 1 and 2 to the
qubits.  They lack a common ``reference ordering'' for
the subsystems \cite{brs}.
If Alice and Bob have both the same subsystem decomposition
and the same reference ordering, we say that they share an
{\em ordered} meronomic frame.  The symmetry group $M_{0}$ for 
such a frame includes all product unitary operators
$\oper{V}\sys{1} \otimes \oper{W}\sys{2}$, but not $\mathbb{X}\sys{12}$.
Schematically,
\begin{equation}
	U(\hilbert) 
	\xrightarrow{\qubitspace \otimes \qubitspace} M
	\xrightarrow{\oper{\xi}} M_{0}
	\xrightarrow{\oper{V} \otimes \oper{W}} \mathds{1} ,
\end{equation}
were $\oper{\xi}$ is either $\oper{1}$ or $\mathbb{X}$.

If Alice and Bob initially share a meronomic frame for tictac systems,
how can they come to share an ordered meronomic frame?  To do
this, Alice prepares a mixed state $\oper{\tau}$ of a tictac pair:
\begin{equation}
	\oper{\tau} = \proj{\Psi_{-}\sys{13}} \otimes \left ( \frac{1}{3} \, \oper{\Pi}_{\mathcal{S}}\sys{24} \right ),
\end{equation}
where $\oper{\Pi}_{\mathcal{S}}$ and $\proj{\Psi_{-}}$
are projections onto symmetric and antisymmetric subspaces of
a two-qubit system.  The state $\oper{\tau}$ is invariant under
the elements of $M_{0}$, but 
\begin{eqnarray}
	\oper{\tau}' & = & (\mathbb{X}\sys{12} \otimes \mathbb{X}\sys{34})
		\oper{\tau} (\mathbb{X}\sys{12} \otimes \mathbb{X}\sys{34})^{\dagger}
		\nonumber \\ & = & 
		\left ( \frac{1}{3} \, \oper{\Pi}_{\mathcal{S}}\sys{13} \right ) \otimes
		\proj{\Psi_{-}\sys{24}} 
\end{eqnarray}
is in fact orthogonal to $\oper{\tau}$.  When Bob receives the tictac
pair from Alice, it is either in state $\oper{\tau}$ or $\oper{\tau}'$,
depending on whether their frames differ by a SWAP.  He can
distinguish these two possibilities by a measurement, and then
adjust his own subsystem labels to agree with Alice's.

\section{Prospects and questions}  \label{sec-questions}

In quantum theory, we often imagine that the
quantum world is assembled from elementary subsystems.
A quantum computer, for instance, is regarded as a collection of
well-defined qubits.  The tensor product rule for system
composition is posed as a basic postulate of quantum theory
\cite{ncbook}.
Even abstract axiomatic reconstructions of quantum theory
often include a composition axiom that leads to the tensor
product rule \cite{hardy,chiribella}.

But the world is not like that.  Subsystem decomposition is a
structure imposed on a system by our description of it.
Of course, our choice of meronomic frame may be motivated by 
sound practical considerations.  We ask what degrees of 
freedom have nearly independent dynamics, or are
individually controllable, or are resistant to decoherence
\cite{lw}.
In any experimental realization of quantum computing, a
key issue is identifying suitable qubits!  Nevertheless,
the meronomic frame we use remains a 
choice, and any composite quantum system can be
decomposed into subsystems in infinitely many ways.

A meronomic frame is a prerequisite for many
concepts in quantum theory.  Without a given division into
subsystems, we cannot begin to discuss entanglement,
decoherence, locality, or a host of other issues.
All of these depend in an essential way on the choice of 
meronomic frame.

In this paper we have only begun to explore meronomic
reference frames, and many questions are still open.
The characterization in Theorem~\ref{thm-products} applies to
bipartite systems of all sizes, and the equivalence of
statements I and III also holds for $n$-partite systems
(with a suitable generalization of the definition of 
the group $M$).  But what about Theorem~\ref{thm-maxents},
which we have proven only for tictac systems?

Furthermore, we observe that the ``partial superdense coding''
task described in Section~\ref{sec-tasks} does not require
a complete shared meronomic frame.  It suffices for Alice
and Bob to agree on a two-dimensional subspace
$\mathcal{P} \subset \tictacspace$ of product states,
from which Bob can construct the required unitary $\oper{V}$.
This suggests that Alice and Bob can share a
``sub-meronomic'' frame, giving them some (but not all)
of the power of a shared subsystem decomposition.
How should such frames be described?  What quantum
resources can represent them?

Here we must acknowledge a difficulty.
In the theory of quantum reference
frames, frame information is embodied in asymmetric physical
states.  The spatial coordinates used in a laboratory,
necessary for experiments on atomic-scale systems, depend
on the translationally and rotationally asymmetric state 
of the laboratory apparatus.  But notice that the whole
idea of a quantum reference frame presupposes a subsystem
decomposition of the world.  This is equally true of our quantum
meronomic frames.  We can use the $\ket{\Lambda}$ state to
represent the subsystem decomposition of a tictac, but only
if we have already divided the world into tictacs.

This raises the spectre of infinite regress.  We can represent 
subsystem decomposition by physical means only if we 
take as given the decomposition of a larger system.  
We must therefore conclude that our discussion of
quantum meronomic frames does not completely account
for the physical origin of the subsystem structure
of the quantum world.

The authors gratefully acknowledge many helpful conversations
about meronomic frames and related questions with Ian George,
Iman Marvian, Roman Plesser, Paul Skrzypczyk, Mike Westmoreland,
and Bill Wootters.

\appendix
\section{Characterizing $M$ for qubit pairs}

In this appendix we prove the following theorem for tictac systems,
which can be viewed as a supplement to Theorem~\ref{thm-products}.
\begin{theorem}  \label{thm-maxents}
	Suppose $\hilbert\sys{12} = \qubitspace\sys{1} \otimes \qubitspace\sys{2}$.  Then
	the following are equivalent:
	\begin{enumerate}[label={\Roman*}.]
	\item  $\oper{U}\sys{12} \in M$.  \setcounter{enumi}{3}
	\item  $\oper{U}\sys{12} \ket{\Psi\sys{12}}$ is maximally entangled if and only if 
		$\ket{\Phi\sys{12}}$ is maximally entangled.  
	\end{enumerate}
\end{theorem}
\begin{proof}
Maximally entangled states of a pair of qubits are defined as those with
Schmidt parameters $\lambda_{0} = \lambda_{1} = 1/2$.
Theorem~\ref{thm-products} therefore tells us that statement
I implies statement IV.  
It remains to show that statement IV implies statement I.

Any maximally entangled state of two qubits can be written
\begin{equation}
	\ket{\Psi} = \frac{1}{\sqrt{2}} \left ( \ket{0,\psi_{0}} + \ket{1,\psi_{1}} \right ),
\end{equation}
where the qubit states $\ket{\psi_{0}}$ and $\ket{\psi_{1}}$ are orthogonal.
It follows that any two maximally entangled states $\ket{\Psi}$ and
$\ket{\Phi}$ are related by
\begin{equation}
	\ket{\Phi} = \idop \otimes \oper{U}^{\phi}_{\psi} \, \ket{\Psi},
\end{equation}
where $\oper{U}^{\phi}_{\psi} = \outerprod{\phi_{0}}{\psi_{0}} + \outerprod{\phi_{1}}{\psi_{1}}$.
We proceed by proving a pair of useful lemmas about this operator.  First we have
\begin{lemma}  \label{lem-antiherm}
Suppose $\ket{\Psi}$ and $\ket{\Phi}$ are maximally entangled states of two
qubits.  Then $\frac{1}{\sqrt{2}} ( \ket{\Psi} + \ket{\Phi} )$ is a maximally
entangled state if and only if $\oper{U}^{\phi}_{\psi}$ is anti-Hermitian.
\end{lemma}
Expanding the superposition, we see that
\begin{eqnarray}
	\frac{1}{\sqrt{2}} \left ( \ket{\Psi} + \ket{\Phi} \right )
	& = & \frac{1}{2} \big( \ket{0} \otimes ( \ket{\psi_{0}} + \ket{\phi_{0}} )
		\nonumber \\ & & \quad
		+ \ket{1} \otimes ( \ket{\psi_{1}} + \ket{\phi_{1}} )
		\big) .
\end{eqnarray}
This is maximally entangled if and only if $\frac{1}{\sqrt{2}} ( \ket{\psi_{0}} + \ket{\phi_{0}} )$
and $\frac{1}{\sqrt{2}} ( \ket{\psi_{1}} + \ket{\phi_{1}} )$ are orthogonal and normalized.
In other words, the following conditions hold:
\begin{itemize}
	\item  $\amp{\psi_{0}}{\phi_{0}} + \amp{\phi_{0}}{\psi_{0}} = 0$.
	\item  $\amp{\psi_{1}}{\phi_{1}} + \amp{\phi_{1}}{\psi_{1}} = 0$.
	\item  $\amp{\psi_{0}}{\phi_{1}} + \amp{\phi_{0}}{\psi_{1}} = 0$.
\end{itemize}
(The first two arise from normaliztion and the last from orthogonality.)
With respect to the $\{ \ket{\phi_{0}}, \ket{\phi_{1}} \}$ basis, the matrix 
form for $\oper{U}^{\phi}_{\psi}$ is
\begin{equation}
\left ( \oper{U}^{\phi}_{\psi} \right )  = 
	\left ( \begin{array}{cc}  \amp{\psi_{0}}{\phi_{0}} & \amp{\psi_{0}}{\phi_{1}} \\
					\amp{\psi_{1}}{\phi_{0}} & \amp{\psi_{1}}{\phi_{1}}
					\end{array} \right )
\end{equation}
Thus, $\oper{U}^{\phi}_{\psi}$ is anti-Hermitian 
(that is, $\oper{U}^{\phi}_{\psi} + {\oper{U}^{\phi}_{\psi}}^{\dagger} = 0$)
if and only if the three conditions hold.

Our second lemma is closely related.
\begin{lemma}  \label{lem-herm}
Suppose $\ket{\Psi}$ and $\ket{\Phi}$ are orthogonal maximally entangled 
states of two qubits.  If $\oper{U}^{\phi}_{\psi}$ is Hermitian then
$\frac{1}{\sqrt{2}} ( \ket{\Psi} + \ket{\Phi} )$ is a product state.
\end{lemma}

The superposition state is
\begin{eqnarray}
	\frac{1}{\sqrt{2}} ( \ket{\Psi} + \ket{\Phi} ) 
	& = & \ket{0} \otimes \frac{1}{2} \left ( \ket{\psi_{0}} + \ket{\phi_{0}} \right )
		\nonumber \\ & & \quad 
		+ \ket{1} \otimes \frac{1}{2} \left ( \ket{\psi_{1}} + \ket{\phi_{1}} \right )
		\nonumber \\
	& = & \ket{0} \otimes \frac{1}{2} \left ( \idop +  \oper{U}^{\phi}_{\psi} \right ) \ket{\phi_{0}}
		\nonumber \\ & & \quad 
		+ \ket{1} \otimes \frac{1}{2} \left ( \idop +  \oper{U}^{\phi}_{\psi} \right ) \ket{\phi_{1}}
		\nonumber \\
	& = & \ket{0} \otimes \ket{\chi_{0}} + \ket{1} \otimes \ket{\chi_{1}} .
\end{eqnarray}
The operator $\oper{U}^{\phi}_{\psi}$ is both unitary and Hermitian, so its eigenvalues
are all +1 or -1.  The operator $\frac{1}{2} ( \idop + \oper{U}^{\phi}_{\psi} )$ has
eigenvalues +1 or 0, and thus is a projection operator.  We know that this is 
neither the identity nor the zero operator, so it must be that 
\begin{equation}
	\frac{1}{2} ( \idop + \oper{U}^{\phi}_{\psi} ) = \proj{u}
\end{equation}
for some qubit state $\ket{u}$.  The vectors $\ket{\chi_{0}}$ and $\ket{\chi_{1}}$
are both multiples of $\ket{u}$, and 
therefore $\frac{1}{\sqrt{2}} ( \ket{\Psi} + \ket{\Phi} )$ is a product state.

We are now ready to prove Theorem~\ref{thm-maxents} itself.  Let
$\ket{\psi,\phi}$ be an arbitrary product state.  We define maximally
entangled states
\begin{eqnarray}
	\ket{\Gamma} 
	& = & \frac{1}{\sqrt{2}} \left ( \ket{\psi,\phi} + \ket{\psi^{\perp},\phi^{\perp}} \right ) \\
	\ket{\Delta} 
	& = & \frac{1}{\sqrt{2}} \left ( \ket{\psi,\phi} - \ket{\psi^{\perp},\phi^{\perp}} \right ) .
\end{eqnarray}
Clearly $\ket{\psi,\phi} = \frac{1}{\sqrt{2}} ( \ket{\Gamma} + \ket{\Delta} )$.  The
operator $\oper{U}^{\delta}_{\gamma} = \proj{\phi} - \proj{\phi^{\perp}}$ is Hermitian, and
so $i \oper{U}^{\delta}_{\gamma}$ is anti-Hermitian.  Therefore, by 
Lemma~\ref{lem-antiherm}, the state
\begin{equation}
	\frac{1}{\sqrt{2}} \left ( \ket{\Gamma} + (\idop \otimes \oper{U}^{\delta}_{\gamma} ) \ket{\Gamma} \right )
		= \frac{1}{\sqrt{2}} \left ( \ket{\Gamma} + i \ket{\Delta} \right )
\end{equation}
is maximally entangled.

The operator $\oper{U}\sys{12}$ is assumed to take maximally entangled
states to maximally entangled states.  Thus, $\ket{\Gamma'} = \oper{U}\sys{12} \ket{\Gamma}$,
$\ket{\Delta'} = \oper{U}\sys{12} \ket{\Delta}$, and $\frac{1}{\sqrt{2}} ( \ket{\Gamma'} + i \ket{\Delta'} )$
are all maximally entangled.  The operator $i \oper{U}^{\delta'}_{\gamma'}$ must be
anti-Hermitian, which in turn implies that $\oper{U}^{\delta'}_{\gamma'}$ is Hermitian.
Lemma~\ref{lem-herm} tells us that the state
\begin{equation}
	\oper{U}\sys{12} \ket{\psi,\phi} = 
	\frac{1}{\sqrt{2}} \left ( \ket{\Gamma'} + \ket{\Delta'} \right )
\end{equation}
must be a product state.  The operator $\oper{U}\sys{12}$ takes product
states to product states, and therefore (by Theorem~\ref{thm-products})
must be in the meronomic subgroup $M$.
\end{proof}

\end{document}